\documentclass[conference, nofonttune]{IEEEtran}
\IEEEoverridecommandlockouts
\usepackage{cite}
\DeclareMathAlphabet{\pazocal}{OMS}{zplm}{m}{n}
\usepackage{graphicx}
\usepackage{textcomp}
\usepackage{xcolor}
\usepackage{amsthm, amsmath,amssymb,amsfonts}
\usepackage{xspace}
\newtheorem{theorem}{Theorem}
\usepackage{pgfplots}
\pgfplotsset{compat=1.10} 
\usepackage{float}
\usepackage{booktabs}
\RequirePackage{xstring}
\RequirePackage{xparse}
\RequirePackage{acro}
\usepackage{url}
\usepackage{subcaption}
\usepackage[capitalise]{cleveref}
\crefname{section}{Sect.}{Sect.}
\crefname{figure}{Figure}{Fig.}
\crefname{table}{Tab.}{Tab.}
\crefname{equation}{Eq.}{Eq.}
\usepackage[letterpaper]{geometry}
\usepackage{todonotes}
\IEEEaftertitletext{\vspace{-10pt}}

\def\BibTeX{{\rm B\kern-.05em{\sc i\kern-.025em b}\kern-.08em
    T\kern-.1667em\lower.7ex\hbox{E}\kern-.125emX}}
\IEEEaftertitletext{\vspace{-1.5\baselineskip}}
\usepackage[acronyms,nonumberlist,nopostdot,nomain,nogroupskip,acronymlists={hidden}]{glossaries}
\newglossary[algh]{hidden}{acrh}{acnh}{Hidden Acronyms}

\newcommand{\Vset}{\ensuremath{{\pazocal V}}\xspace}
\newcommand{\Eset}{\ensuremath{{\pazocal E}}\xspace}
\newcommand{\Ggraph}{\ensuremath{{\pazocal G}}\xspace}

\newcommand{\lambdagnb}{\ensuremath{\lambda_{\text{\gls{gnb}}}}\xspace}
\newcommand{\lambdaue}{\ensuremath{\lambda_{\text{\gls{ue}}}}\xspace}

\newacronym{3gpp}{3GPP}{3rd Generation Partnership Project}
\newacronym{4g}{4G}{4th generation}
\newacronym{5g}{5G}{5th generation}
\newacronym{6g}{6G}{6th generation}
\newacronym{5gc}{5GC}{5G Core}
\newacronym{adc}{ADC}{Analog to Digital Converter}
\newacronym{aerpaw}{AERPAW}{Aerial Experimentation and Research Platform for Advanced Wireless}
\newacronym{ai}{AI}{Artificial Intelligence}
\newacronym{aimd}{AIMD}{Additive Increase Multiplicative Decrease}
\newacronym{am}{AM}{Acknowledged Mode}
\newacronym{amc}{AMC}{Adaptive Modulation and Coding}
\newacronym{amf}{AMF}{Access and Mobility Management Function}
\newacronym{aops}{AOPS}{Adaptive Order Prediction Scheduling}
\newacronym{api}{API}{Application Programming Interface}
\newacronym{apn}{APN}{Access Point Name}
\newacronym{ap}{AP}{Application Protocol}
\newacronym{aqm}{AQM}{Active Queue Management}
\newacronym{ausf}{AUSF}{Authentication Server Function}
\newacronym{avc}{AVC}{Advanced Video Coding}
\newacronym{awgn}{AGWN}{Additive White Gaussian Noise}
\newacronym{balia}{BALIA}{Balanced Link Adaptation Algorithm}
\newacronym{bbu}{BBU}{Base Band Unit}
\newacronym{bdp}{BDP}{Bandwidth-Delay Product}
\newacronym{ber}{BER}{Bit Error Rate}
\newacronym{bf}{BF}{Beamforming}
\newacronym{bler}{BLER}{Block Error Rate}
\newacronym{brr}{BRR}{Bayesian Ridge Regressor}
\newacronym{bs}{BS}{Base Station}
\newacronym{bsr}{BSR}{Buffer Status Report}
\newacronym{bss}{BSS}{Business Support System}
\newacronym{ca}{CA}{Carrier Aggregation}
\newacronym{caas}{CaaS}{Connectivity-as-a-Service}
\newacronym{cb}{CB}{Code Block}
\newacronym{cc}{CC}{Congestion Control}
\newacronym{ccid}{CCID}{Congestion Control ID}
\newacronym{cco}{CC}{Carrier Component}
\newacronym{cd}{CD}{Continuous Delivery}
\newacronym{cdd}{CDD}{Cyclic Delay Diversity}
\newacronym{cdf}{CDF}{Cumulative Distribution Function}
\newacronym{cdn}{CDN}{Content Distribution Network}
\newacronym{cli}{CLI}{Command-line Interface}
\newacronym{cn}{CN}{Core Network}
\newacronym{codel}{CoDel}{Controlled Delay Management}
\newacronym{comac}{COMAC}{Converged Multi-Access and Core}
\newacronym{cord}{CORD}{Central Office Re-architected as a Datacenter}
\newacronym{cornet}{CORNET}{COgnitive Radio NETwork}
\newacronym{cosmos}{COSMOS}{Cloud Enhanced Open Software Defined Mobile Wireless Testbed for City-Scale Deployment}
\newacronym{cots}{COTS}{Commercial Off-the-Shelf}
\newacronym{cp}{CP}{Control Plane}
\newacronym{cyp}{CP}{Cyclic Prefix}
\newacronym{up}{UP}{User Plane}
\newacronym{cpu}{CPU}{Central Processing Unit}
\newacronym{cqi}{CQI}{Channel Quality Information}
\newacronym{cr}{CR}{Cognitive Radio}
\newacronym{cran}{CRAN}{Cloud \gls{ran}}
\newacronym{crs}{CRS}{Cell Reference Signal}
\newacronym{csi}{CSI}{Channel State Information}
\newacronym{csirs}{CSI-RS}{Channel State Information - Reference Signal}
\newacronym{cu}{CU}{Central Unit}
\newacronym{d2tcp}{D$^2$TCP}{Deadline-aware Data center TCP}
\newacronym{d3}{D$^3$}{Deadline-Driven Delivery}
\newacronym{dac}{DAC}{Digital to Analog Converter}
\newacronym{dag}{DAG}{Directed Acyclic Graph}
\newacronym{das}{DAS}{Distributed Antenna System}
\newacronym{dash}{DASH}{Dynamic Adaptive Streaming over HTTP}
\newacronym{dc}{DC}{Dual Connectivity}
\newacronym{dccp}{DCCP}{Datagram Congestion Control Protocol}
\newacronym{dce}{DCE}{Direct Code Execution}
\newacronym{dci}{DCI}{Downlink Control Information}
\newacronym{dctcp}{DCTCP}{Data Center TCP}
\newacronym{dl}{DL}{Downlink}
\newacronym{dmr}{DMR}{Deadline Miss Ratio}
\newacronym{dmrs}{DMRS}{DeModulation Reference Signal}
\newacronym{drlcc}{DRL-CC}{Deep Reinforcement Learning Congestion Control}
\newacronym{drs}{DRS}{Discovery Reference Signal}
\newacronym{du}{DU}{Distributed Unit}
\newacronym{e2e}{E2E}{end-to-end}
\newacronym{earfcn}{EARFCN}{E-UTRA Absolute Radio Frequency Channel Number}
\newacronym{ecaas}{ECaaS}{Edge-Cloud-as-a-Service}
\newacronym{ecn}{ECN}{Explicit Congestion Notification}
\newacronym{edf}{EDF}{Earliest Deadline First}
\newacronym{embb}{eMBB}{Enhanced Mobile Broadband}
\newacronym{empower}{EMPOWER}{EMpowering transatlantic PlatfOrms for advanced WirEless Research}
\newacronym{enb}{eNB}{evolved Node Base}
\newacronym{endc}{EN-DC}{E-UTRAN-\gls{nr} \gls{dc}}
\newacronym{epc}{EPC}{Evolved Packet Core}
\newacronym{eps}{EPS}{Evolved Packet System}
\newacronym{es}{ES}{Edge Server}
\newacronym{etsi}{ETSI}{European Telecommunications Standards Institute}
\newacronym[firstplural=Estimated Times of Arrival (ETAs)]{eta}{ETA}{Estimated Time of Arrival}
\newacronym{eutran}{E-UTRAN}{Evolved Universal Terrestrial Access Network}
\newacronym{faas}{FaaS}{Function-as-a-Service}
\newacronym{fapi}{FAPI}{Functional Application Platform Interface}
\newacronym{fdd}{FDD}{Frequency Division Duplexing}
\newacronym{fdm}{FDM}{Frequency Division Multiplexing}
\newacronym{fdma}{FDMA}{Frequency Division Multiple Access}
\newacronym{fed4fire}{FED4FIRE+}{Federation 4 Future Internet Research and Experimentation Plus}
\newacronym{fir}{FIR}{Finite Impulse Response}
\newacronym{fit}{FIT}{Future \acrlong{iot}}
\newacronym{fpga}{FPGA}{Field Programmable Gate Array}
\newacronym{fr2}{FR2}{Frequency Range 2}
\newacronym{fs}{FS}{Fast Switching}
\newacronym{fscc}{FSCC}{Flow Sharing Congestion Control}
\newacronym{ftp}{FTP}{File Transfer Protocol}
\newacronym{fw}{FW}{Flow Window}
\newacronym{ge}{GE}{Gaussian Elimination}
\newacronym{gnb}{gNB}{Next Generation Node Base}
\newacronym{gop}{GOP}{Group of Pictures}
\newacronym{gpr}{GPR}{Gaussian Process Regressor}
\newacronym{gpu}{GPU}{Graphics Processing Unit}
\newacronym{gtp}{GTP}{GPRS Tunneling Protocol}
\newacronym{gtpc}{GTP-C}{GPRS Tunnelling Protocol Control Plane}
\newacronym{gtpu}{GTP-U}{GPRS Tunnelling Protocol User Plane}
\newacronym{gtpv2c}{GTPv2-C}{\gls{gtp} v2 - Control}
\newacronym{gw}{GW}{Gateway}
\newacronym{harq}{HARQ}{Hybrid Automatic Repeat reQuest}
\newacronym{hetnet}{HetNet}{Heterogeneous Network}
\newacronym{hh}{HH}{Hard Handover}
\newacronym{hol}{HOL}{Head-of-Line}
\newacronym{hqf}{HQF}{Highest-quality-first}
\newacronym{hss}{HSS}{Home Subscription Server}
\newacronym{http}{HTTP}{HyperText Transfer Protocol}
\newacronym{ia}{IA}{Initial Access}
\newacronym{iab}{IAB}{Integrated Access and Backhaul}
\newacronym{ic}{IC}{Incident Command}
\newacronym{ietf}{IETF}{Internet Engineering Task Force}
\newacronym{imsi}{IMSI}{International Mobile Subscriber Identity}
\newacronym{imt}{IMT}{International Mobile Telecommunication}
\newacronym{iot}{IoT}{Internet of Things}
\newacronym{ip}{IP}{Internet Protocol}
\newacronym{itu}{ITU}{International Telecommunication Union}
\newacronym{kpi}{KPI}{Key Performance Indicator}
\newacronym{kpm}{KPM}{Key Performance Measurement}
\newacronym{kvm}{KVM}{Kernel-based Virtual Machine}
\newacronym{los}{LoS}{Line of Sight}
\newacronym{lsm}{LSM}{Link-to-System Mapping}
\newacronym{lstm}{LSTM}{Long Short Term Memory}
\newacronym{lte}{LTE}{Long Term Evolution}
\newacronym{lxc}{LXC}{Linux Container}
\newacronym{m2m}{M2M}{Machine to Machine}
\newacronym{mac}{MAC}{Medium Access Control}
\newacronym{manet}{MANET}{Mobile Ad Hoc Network}
\newacronym{mano}{MANO}{Management and Orchestration}
\newacronym{mc}{MC}{Multi-Connectivity}
\newacronym{mcc}{MCC}{Mobile Cloud Computing}
\newacronym{mchem}{MCHEM}{Massive Channel Emulator}
\newacronym{mcs}{MCS}{Modulation and Coding Scheme}
\newacronym{mec2}{MEC}{Multi-access Edge Computing}
\newacronym{mec}{MEC}{Mobile Edge Computing}
\newacronym{mfc}{MFC}{Mobile Fog Computing}
\newacronym{mgen}{MGEN}{Multi-Generator}
\newacronym{mi}{MI}{Mutual Information}
\newacronym{mib}{MIB}{Master Information Block}
\newacronym{miesm}{MIESM}{Mutual Information Based Effective SINR}
\newacronym{mimo}{MIMO}{Multiple Input, Multiple Output}
\newacronym{ml}{ML}{Machine Learning}
\newacronym{mlr}{MLR}{Maximum-local-rate}
\newacronym[plural=\gls{mme}s,firstplural=Mobility Management Entities (MMEs)]{mme}{MME}{Mobility Management Entity}
\newacronym{mmtc}{mMTC}{Massive Machine-Type Communications}
\newacronym{mmwave}{mmWave}{millimeter wave}
\newacronym{mpdccp}{MP-DCCP}{Multipath Datagram Congestion Control Protocol}
\newacronym{mptcp}{MPTCP}{Multipath TCP}
\newacronym{mr}{MR}{Maximum Rate}
\newacronym{mrdc}{MR-DC}{Multi \gls{rat} \gls{dc}}
\newacronym{mse}{MSE}{Mean Square Error}
\newacronym{mss}{MSS}{Maximum Segment Size}
\newacronym{mt}{MT}{Mobile Termination}
\newacronym{mtd}{MTD}{Machine-Type Device}
\newacronym{mtu}{MTU}{Maximum Transmission Unit}
\newacronym{mumimo}{MU-MIMO}{Multi-user \gls{mimo}}
\newacronym{mvno}{MVNO}{Mobile Virtual Network Operator}
\newacronym{nalu}{NALU}{Network Abstraction Layer Unit}
\newacronym{nas}{NAS}{Network Attached Storage}
\newacronym{nat}{NAT}{Network Address Translation}
\newacronym{nbiot}{NB-IoT}{Narrow Band IoT}
\newacronym{nfv}{NFV}{Network Function Virtualization}
\newacronym{nfvi}{NFVI}{Network Function Virtualization Infrastructure}
\newacronym{ni}{NI}{Network Interfaces}
\newacronym{nic}{NIC}{Network Interface Card}
\newacronym{now}{NOW}{Non Overlapping Window}
\newacronym{nsm}{NSM}{Network Service Mesh}
\newacronym{nr}{NR}{New Radio}
\newacronym{nrf}{NRF}{Network Repository Function}
\newacronym{nsa}{NSA}{Non Stand Alone}
\newacronym{nse}{NSE}{Network Slicing Engine}
\newacronym{nssf}{NSSF}{Network Slice Selection Function}
\newacronym{o2i}{O2I}{Outdoor to Indoor}
\newacronym{oai}{OAI}{OpenAirInterface}
\newacronym{oaicn}{OAI-CN}{\gls{oai} \acrlong{cn}}
\newacronym{oairan}{OAI-RAN}{\acrlong{oai} \acrlong{ran}}
\newacronym{oam}{OAM}{Operations, Administration and Maintenance}
\newacronym{ofdm}{OFDM}{Orthogonal Frequency Division Multiplexing}
\newacronym{olia}{OLIA}{Opportunistic Linked Increase Algorithm}
\newacronym{omec}{OMEC}{Open Mobile Evolved Core}
\newacronym{onap}{ONAP}{Open Network Automation Platform}
\newacronym{onf}{ONF}{Open Networking Foundation}
\newacronym{onos}{ONOS}{Open Networking Operating System}
\newacronym{oom}{OOM}{\gls{onap} Operations Manager}
\newacronym{opnfv}{OPNFV}{Open Platform for \gls{nfv}}
\newacronym{oran}{O-RAN}{Open Radio Access Network}
\newacronym{orbit}{ORBIT}{Open-Access Research Testbed for Next-Generation Wireless Networks}
\newacronym{os}{OS}{Operating System}
\newacronym{oss}{OSS}{Operations Support System}
\newacronym{pa}{PA}{Position-aware}
\newacronym{pase}{PASE}{Prioritization, Arbitration, and Self-adjusting Endpoints}
\newacronym{pawr}{PAWR}{Platforms for Advanced Wireless Research}
\newacronym{pbch}{PBCH}{Physical Broadcast Channel}
\newacronym{pcef}{PCEF}{Policy and Charging Enforcement Function}
\newacronym{pcfich}{PCFICH}{Physical Control Format Indicator Channel}
\newacronym{pcrf}{PCRF}{Policy and Charging Rules Function}
\newacronym{pdcch}{PDCCH}{Physical Downlink Control Channel}
\newacronym{pdcp}{PDCP}{Packet Data Convergence Protocol}
\newacronym{pdsch}{PDSCH}{Physical Downlink Shared Channel}
\newacronym{pdu}{PDU}{Packet Data Unit}
\newacronym{pf}{PF}{Proportional Fair}
\newacronym{pgw}{PGW}{Packet Gateway}
\newacronym{phich}{PHICH}{Physical Hybrid ARQ Indicator Channel}
\newacronym{phy}{PHY}{Physical}
\newacronym{pmch}{PMCH}{Physical Multicast Channel}
\newacronym{pmi}{PMI}{Precoding Matrix Indicators}
\newacronym{powder}{POWDER}{Platform for Open Wireless Data-driven Experimental Research}
\newacronym{ppo}{PPO}{Proximal Policy Optimization}
\newacronym{ppp}{PPP}{Poisson Point Process}
\newacronym{prach}{PRACH}{Physical Random Access Channel}
\newacronym{prb}{PRB}{Physical Resource Block}
\newacronym{psnr}{PSNR}{Peak Signal to Noise Ratio}
\newacronym{pss}{PSS}{Primary Synchronization Signal}
\newacronym{pucch}{PUCCH}{Physical Uplink Control Channel}
\newacronym{pusch}{PUSCH}{Physical Uplink Shared Channel}
\newacronym{qam}{QAM}{Quadrature Amplitude Modulation}
\newacronym{qci}{QCI}{\gls{qos} Class Identifier}
\newacronym{qoe}{QoE}{Quality of Experience}
\newacronym{qos}{QoS}{Quality of Service}
\newacronym{quic}{QUIC}{Quick UDP Internet Connections}
\newacronym{rach}{RACH}{Random Access Channel}
\newacronym{ran}{RAN}{Radio Access Network}
\newacronym[firstplural=Radio Access Technologies (RATs)]{rat}{RAT}{Radio Access Technology}
\newacronym{rbg}{RBG}{Resource Block Group}
\newacronym{rcn}{RCN}{Research Coordination Network}
\newacronym{rc}{RC}{RAN Control}
\newacronym{rec}{REC}{Radio Edge Cloud}
\newacronym{red}{RED}{Random Early Detection}
\newacronym{renew}{RENEW}{Reconfigurable Eco-system for Next-generation End-to-end Wireless}
\newacronym{rf}{RF}{Radio Frequency}
\newacronym{rfc}{RFC}{Request for Comments}
\newacronym{rfr}{RFR}{Random Forest Regressor}
\newacronym{ric}{RIC}{\gls{ran} Intelligent Controller}
\newacronym{rlc}{RLC}{Radio Link Control}
\newacronym{rlf}{RLF}{Radio Link Failure}
\newacronym{rlnc}{RLNC}{Random Linear Network Coding}
\newacronym{rmr}{RMR}{RIC Message Router}
\newacronym{rmse}{RMSE}{Root Mean Squared Error}
\newacronym{rnis}{RNIS}{Radio Network Information Service}
\newacronym{rr}{RR}{Round Robin}
\newacronym{rrc}{RRC}{Radio Resource Control}
\newacronym{rrm}{RRM}{Radio Resource Management}
\newacronym{rru}{RRU}{Remote Radio Unit}
\newacronym{rs}{RS}{Remote Server}
\newacronym{rsrp}{RSRP}{Reference Signal Received Power}
\newacronym{rsrq}{RSRQ}{Reference Signal Received Quality}
\newacronym{rss}{RSS}{Received Signal Strength}
\newacronym{rssi}{RSSI}{Received Signal Strength Indicator}
\newacronym{rtt}{RTT}{Round Trip Time}
\newacronym{ru}{RU}{Radio Unit}
\newacronym{rw}{RW}{Receive Window}
\newacronym{rx}{RX}{Receiver}
\newacronym{s1ap}{S1AP}{S1 Application Protocol}
\newacronym{sa}{SA}{standalone}
\newacronym{sack}{SACK}{Selective Acknowledgment}
\newacronym{sap}{SAP}{Service Access Point}
\newacronym{sc2}{SC2}{Spectrum Collaboration Challenge}
\newacronym{scef}{SCEF}{Service Capability Exposure Function}
\newacronym{sch}{SCH}{Secondary Cell Handover}
\newacronym{scoot}{SCOOT}{Split Cycle Offset Optimization Technique}
\newacronym{sctp}{SCTP}{Stream Control Transmission Protocol}
\newacronym{sdap}{SDAP}{Service Data Adaptation Protocol}
\newacronym{sdk}{SDK}{Software Development Kit}
\newacronym{sdm}{SDM}{Space Division Multiplexing}
\newacronym{sdma}{SDMA}{Spatial Division Multiple Access}
\newacronym{sdn}{SDN}{Software-defined Networking}
\newacronym{sdr}{SDR}{Software-defined Radio}
\newacronym{seba}{SEBA}{SDN-Enabled Broadband Access}
\newacronym{sgsn}{SGSN}{Serving GPRS Support Node}
\newacronym{sgw}{SGW}{Service Gateway}
\newacronym{si}{SI}{Study Item}
\newacronym{sib}{SIB}{Secondary Information Block}
\newacronym{sinr}{SINR}{Signal to Interference plus Noise Ratio}
\newacronym{sip}{SIP}{Session Initiation Protocol}
\newacronym{siso}{SISO}{Single Input, Single Output}
\newacronym{sla}{SLA}{Service Level Agreement}
\newacronym{sm}{SM}{Service Model}
\newacronym{smo}{SMO}{Service Management and Orchestration}
\newacronym{smsgmsc}{SMS-GMSC}{\gls{sms}-Gateway}
\newacronym{snr}{SNR}{Signal-to-Noise-Ratio}
\newacronym{son}{SON}{Self-Organizing Network}
\newacronym{sptcp}{SPTCP}{Single Path TCP}
\newacronym{srb}{SRB}{Service Radio Bearer}
\newacronym{srn}{SRN}{Standard Radio Node}
\newacronym{srs}{SRS}{Sounding Reference Signal}
\newacronym{ss}{SS}{Synchronization Signal}
\newacronym{sss}{SSS}{Secondary Synchronization Signal}
\newacronym{st}{ST}{Spanning Tree}
\newacronym{svc}{SVC}{Scalable Video Coding}
\newacronym{tb}{TB}{Transport Block}
\newacronym{tcp}{TCP}{Transmission Control Protocol}
\newacronym{tdd}{TDD}{Time Division Duplexing}
\newacronym{tdm}{TDM}{Time Division Multiplexing}
\newacronym{tdma}{TDMA}{Time Division Multiple Access}
\newacronym{tfl}{TfL}{Transport for London}
\newacronym{tfrc}{TFRC}{TCP-Friendly Rate Control}
\newacronym{tft}{TFT}{Traffic Flow Template}
\newacronym{tgen}{TGEN}{Traffic Generator}
\newacronym{tip}{TIP}{Telecom Infra Project}
\newacronym{tm}{TM}{Transparent Mode}
\newacronym{to}{TO}{Telco Operator}
\newacronym{tr}{TR}{Technical Report}
\newacronym{trp}{TRP}{Transmitter Receiver Pair}
\newacronym{ts}{TS}{Technical Specification}
\newacronym{tti}{TTI}{Transmission Time Interval}
\newacronym{ttt}{TTT}{Time-to-Trigger}
\newacronym{tx}{TX}{Transmitter}
\newacronym{uas}{UAS}{Unmanned Aerial System}
\newacronym{uav}{UAV}{Unmanned Aerial Vehicle}
\newacronym{udm}{UDM}{Unified Data Management}
\newacronym{udp}{UDP}{User Datagram Protocol}
\newacronym{udr}{UDR}{Unified Data Repository}
\newacronym{ue}{UE}{User Equipment}
\newacronym{uhd}{UHD}{\gls{usrp} Hardware Driver}
\newacronym{ul}{UL}{Uplink}
\newacronym{um}{UM}{Unacknowledged Mode}
\newacronym{uml}{UML}{Unified Modeling Language}
\newacronym{upa}{UPA}{Uniform Planar Array}
\newacronym{upf}{UPF}{User Plane Function}
\newacronym{urllc}{URLLC}{Ultra Reliable and Low Latency Communications}
\newacronym{usa}{U.S.}{United States}
\newacronym{usim}{USIM}{Universal Subscriber Identity Module}
\newacronym{usrp}{USRP}{Universal Software Radio Peripheral}
\newacronym{utc}{UTC}{Urban Traffic Control}
\newacronym{vim}{VIM}{Virtualization Infrastructure Manager}
\newacronym{vm}{VM}{Virtual Machine}
\newacronym{vnf}{VNF}{Virtual Network Function}
\newacronym{volte}{VoLTE}{Voice over \gls{lte}}
\newacronym{voltha}{VOLTHA}{Virtual OLT HArdware Abstraction}
\newacronym{vr}{VR}{Virtual Reality}
\newacronym{vran}{vRAN}{Virtualized \gls{ran}}
\newacronym{vss}{VSS}{Video Streaming Server}
\newacronym{wbf}{WBF}{Wired Bias Function}
\newacronym{wf}{WF}{Waterfilling}
\newacronym{wg}{WG}{Working Group}
\newacronym{wlan}{WLAN}{Wireless Local Area Network}
\newacronym{osm}{OSM}{Open Source \gls{nfv} Management and Orchestration}
\newacronym{pnf}{PNF}{Physical Network Function}
\newacronym{drl}{DRL}{Deep Reinforcement Learning}
\newacronym{mtc}{MTC}{Machine-type Communications}
\newacronym{osc}{OSC}{O-RAN Software Community}
\newacronym{mns}{MnS}{Management Services}
\newacronym{ves}{VES}{\gls{vnf} Event Stream}
\newacronym{ei}{EI}{Enrichment Information}
\newacronym{fh}{FH}{Fronthaul}
\newacronym{fft}{FFT}{Fast Fourier Transform}
\newacronym{laa}{LAA}{Licensed-Assisted Access}
\newacronym{plfs}{PLFS}{Physical Layer Frequency Signals}
\newacronym{ptp}{PTP}{Precision Time Protocol}
\newacronym{lidar}{LiDAR}{Light Detection And Ranging}
\newacronym{dem}{DEM}{Digital Elevation Model}
\newacronym{dtm}{DEM}{Digital Terrain Model}
\newacronym{dsm}{DEM}{Digital Surface Models}
\newacronym{ota}{OTA}{Over-The-Air}
\newacronym{ns}{NS}{Network Slicing}
\newacronym{ne}{NE}{Nash Equilibrium}
\newacronym{hf}{HF}{High Frequency}
\newacronym{noma}{NOMA}{Non-Orthogonal Multiple Access}
\newacronym{sre}{SRE}{Smart Radio Environment}
\newacronym{ris}{RIS}{Reconfigurable Intelligent Surface}
\newacronym{inp}{InP}{Infrastructure Provider}
\newacronym{smf}{SMF}{Slicing Magangement Framework}
\newacronym{nsn}{NSN}{Network Slicing Negotiation}
\newacronym{sms}{SMS}{Slicing MAC Scheduler}
\newacronym{brd}{BRD}{Best Response Dynamics}
\newacronym{dssbr}{DSSBR}{Double Step Smoothed Best Response}
\newacronym{poa}{PoA}{Price of Anarchy}
\newacronym{pos}{PoS}{Price of Stability}
\newacronym{milp}{MILP}{Mixed Integer-Linear Program}
\newacronym{pod}{PoD}{Price of DSSBR}
\newacronym{roc}{ROC}{Radio Overload Control}
\newacronym{ciot}{cIoT}{critical Internet of Things}
\newacronym{embbpr}{eMBB Pr.}{enhanced Mobile BroadBand Premium}
\newacronym{embbbs}{eMBB Bs.}{enhanced Mobile BroadBand Basic}
\newacronym{en}{EN}{Edge Node}
\newacronym{ec}{EC}{Edge Computing}
\newacronym{sp}{SP}{Service Provider}
\newacronym{me}{ME}{Market Equilibrium}
\newacronym{so}{SO}{Social Optimum}
\newacronym{wso}{WSO}{Weighted Social Optimum}
\newacronym{wsn}{WSN}{Wireless Sensor Network}
\newacronym{ps}{PS}{Proportional Sharing}
\newacronym{eg}{EG}{Eisenberg-Gale program}
\newacronym{pe}{PE}{Pareto Efficiency}
\newacronym{nsw}{NSW}{Nash Social Welfare}
\newacronym{ef}{EF}{Envy-Freeness}
\newacronym{sub6}{sub6GHz}{Below 6GHz}
\newacronym{ncr}{NCR}{Network-Controlled Repeater}
\newacronym{nlos}{NLoS}{Non-LoS}
\newacronym{src}{SRC}{Smart Radio Connection}
\newacronym{srd}{SRD}{Smart Radio Device}
\newacronym{cs}{CS}{Candidate Site}
\newacronym{tp}{TP}{Test Point}
\newacronym{fov}{FoV}{Field of View}
\newacronym{nrric}{Near-RT RIC}{Near Real-time RAN Intelligent Controller}
\newacronym{e2ap}{E2AP}{E2 Application Protocol}
\newacronym{e2sm}{E2SM}{E2 Service Model}
\newacronym{nrtric}{Non-RT RIC}{Non-Real-Time Ran Intelligent Controller}
\newacronym{itti}{ITTI}{Inter-task Interface}
\newacronym{bap}{BAP}{Backhaul Adaptation Protocol}
\newacronym{iabest}{IABEST}{Integrated Access and Backhaul Experimental large-Scale Tetbed}
\newacronym{teid}{TEID}{Tunnel Endpoint Identifier}
\newacronym{dlsch}{DL-SCH}{Downlink Shared Channel }
\newacronym{ulsch}{UL-SCH}{Uplink Shared Channel }
\newacronym{opex}{OpEx}{Operational Expenditure}
\newacronym{capex}{CapEx}{Capital Expenditure}
\newacronym{mno}{MNO}{Mobile Network Operator}

\glsdisablehyper
\begin{document}
%set bst control to limit 3 authors (2 + et al)
\bstctlcite{IEEEexample:BSTcontrol}

\title{Joint Routing and Energy Optimization for\\Integrated Access and Backhaul with Open RAN}

\author{
    \IEEEauthorblockN{Gabriele Gemmi\IEEEauthorrefmark{1}\textsuperscript{\textsection}, Maxime Elkael\IEEEauthorrefmark{2}\textsuperscript{\textsection}, Michele Polese\IEEEauthorrefmark{3}, Leonardo Maccari\IEEEauthorrefmark{1}, Hind Castel-Taleb\IEEEauthorrefmark{2}, Tommaso Melodia\IEEEauthorrefmark{3}}\\
%     \IEEEauthorblockA{\IEEEauthorrefmark{1}Department of
% Environmental Sciences, Informatics and Statistics,  Ca' Foscari University of Venice, Italy.
%     \\\{name.surname\}@unive.it}\vspace{0.5em}
%     \IEEEauthorblockA{\IEEEauthorrefmark{2}SAMOVAR, Telecom Sud-Paris, Institut Polytechnique de Paris, France
%     \\\{name\_surname\}@telecom-sudparis.eu}\vspace{0.5em}
%     \IEEEauthorblockA{\IEEEauthorrefmark{3}Institute for the Wireless Internet of Things, Northeastern University, Boston, MA, U.S.A.
%     \\\{n.surname\}@northeastern.edu}

\IEEEauthorblockA{\IEEEauthorrefmark{1}Department of
Environmental Sciences, Informatics and Statistics,  Ca' Foscari University of Venice, Italy.}
    \IEEEauthorblockA{\IEEEauthorrefmark{2}SAMOVAR, Telecom Sud-Paris, Institut Polytechnique de Paris, France}
    \IEEEauthorblockA{\IEEEauthorrefmark{3}Institute for the Wireless Internet of Things, Northeastern University, Boston, MA, U.S.A.
    \\gabriele.gemmi@unive.it}
    
    \thanks{This work was partially supported by Agence Nationale pour la Recherche through the AIDY-F2N project, grant number ANR-19-LCV2-0012, and by OUSD(R\&E) through Army Research Laboratory Cooperative Agreement Number W911NF-19-2-0221. The views and conclusions contained in this document are those of the authors and should not be interpreted as representing the official policies, either expressed or implied, of the Army Research Laboratory or the U.S. Government. The U.S. Government is authorized to reproduce and distribute reprints for Government purposes notwithstanding any copyright notation herein.}
}

\maketitle
\begingroup\renewcommand\thefootnote{\textsection}
\footnotetext{Equal contribution}

\begin{abstract}
Energy consumption represents a major part of the operating expenses of mobile network operators. With the densification foreseen with 5G and beyond, energy optimization has become a problem of crucial importance. While energy optimization is widely studied in the literature, there are limited insights and algorithms for energy-saving techniques for Integrated Access and Backhaul (IAB), a self-backhauling architecture that ease deployment of dense cellular networks reducing the number of fiber drops. This paper proposes a novel optimization model for dynamic joint routing and energy optimization in IAB networks. We leverage the closed-loop control framework introduced by the Open Radio Access Network (O-RAN) architecture to minimize the number of active IAB nodes while maintaining a minimum capacity per User Equipment (UE). The proposed approach formulates the problem as a binary nonlinear program, which is transformed into an equivalent binary linear program and solved using the Gurobi solver. The approach is evaluated on a scenario built upon open data of two months of traffic collected by network operators in the city of Milan, Italy. Results show that the proposed optimization model reduces the RAN energy consumption by 47\%, while guaranteeing a minimum capacity for each UE.
\end{abstract}


\begin{picture}(0,0)(10,-440)
    \put(0,0){
    \put(0,0){\footnotesize \scshape This paper has been accepted for publication on IEEE GLOBECOM 2023 Global Communications Conference.}
     \put(0,-10){
     \scriptsize\scshape \textcopyright~2023 IEEE. Personal use of this material is permitted. Permission from IEEE must be obtained for all other uses, in any current or future media, including}
     \put(0, -17){
     \scriptsize\scshape reprinting/republishing this material for advertising or promotional purposes, creating new collective works, for resale or redistribution to servers or}
     \put(0, -24){
     \scriptsize\scshape lists, or reuse of any copyrighted component of this work in other works.}
     }
 \end{picture}

\section{Introduction}\label{sect:intro}
Ultra-dense deployment and \gls{mmwave} have been portrayed as the solution to meet the stringent requirements standardized with 5G in terms of data rates \cite{perez2015towards}. As actually proven by many studies, \gls{mmwave} is capable of providing multi-gigabit connectivity to \glspl{ue} \cite{rappaport2013millimeter, akdeniz2014millimeter}, and \gls{iab} has been proven to be an effective way to reduce the deployment costs \cite{polese2020integrated}.
This technology, introduced in 3GPP Release 16, allows, in fact, connecting only a subset of the \glspl{gnb}, called \gls{iab}-donors, to the fiber backhaul while the rest of the \glspl{gnb}, called \gls{iab}-node, rely on in-band wireless communication to reach one of the donors, forming a multihop wireless network.

By dynamically activating and deactivating \glspl{gnb} with respect to the current load of the network, it is possible to reduce the energy footprint of the system, switching off \gls{iab}-nodes that are not strictly necessary to match the requested level of service. This is of critical importance since energy consumption accounts for up to 60\% of the \gls{opex} \cite{gsma2020energy}.  Research on energy optimization techniques for traditional networks usually assumes that all \glspl{gnb} are connected to the fiber backhaul. The presence of wireless-only \gls{iab}-nodes, however, increases the complexity of the scenario. 
%makes a policy for activating and deactivating nodes becomes of critical importance. 
In fact, the deactivation of an \gls{iab}-node might disrupt the service of another \gls{iab}-node that is relying on it. Moreover, \gls{iab}-nodes need to support periodic wake-up to update their radio state in order to dynamically change the topology when needed.

Thanks to the effort led by the \gls{oran} Alliance, which has opened the \gls{ran} architecture by introducing interfaces such as the E2 and O1 and the concept of \gls{ric} \cite{polese2023understanding}, it is now possible to integrate custom closed-loop control logic in the \gls{ran}. This has been already studied for several applications, such as network slicing \cite{doro2022orchestran}, and, more recently, an extension of the O-RAN architecture has been proposed for \gls{iab} \cite{moro2023toward}.

Our novel optimization approach takes advantage of this closed-loop control framework to overcome the limitations discussed above and dynamically minimize the number of active \gls{iab}-nodes, while maintaining a minimum capacity per \gls{ue}. The optimization---based on input data that can be obtained through O-RAN interfaces---generates a topology tree over which we route the traffic from each \gls{ue} to the \gls{iab}-donor, deactivating \gls{iab}-nodes that are not needed and distributing \glspl{ue} across the available \glspl{gnb}.

To the best of our knowledge, joint routing and energy optimization on multi-hop \gls{iab} topologies has never been studied before. Most studies focus on optimizing the \gls{iab} topology with different constraints, such as in \cite{IslamIAB2017} where the fiber-deployment cost is minimized or in \cite{Islam2018} where the \gls{ue} data rates are maximized. 
Other studies, more focused on power and energy-related optimization of \gls{iab} networks, exist, but they either optimize the energy consumption after the topology  and routing have been chosen, such as in \cite{meng2018energy}, or they are restricted to a single-hop architecture, such as in \cite{lei2018noma}, where a low energy multiple access scheme is designed.

Optimization for energy consumption is a common concern in both this study and \glspl{wsn}. However, the goals differ: \glspl{wsn} aim to extend network lifetime due to their battery-powered nodes, as in \cite{rault2014energy}, while our grid-powered nodes target different optimization criteria.

Hence, in this paper, we fill this gap by formulating the problem as a binary nonlinear program. Since its continuous relaxation is non-convex, it is not solvable using off-the-shelf solvers, therefore we show it can be transformed into an equivalent binary linear program, which we then solve using the Gurobi solver. 
The approach is evaluated on a scenario built upon open data of two months of traffic collected by network operators in the city of Milan, Italy, \cite{barlacchi2015multi} together with detailed morphological data of the same area. Our optimization model manages to perfectly tune the number of active \glspl{gnb} with the number of \gls{ue}, reducing by 47\% the total number of hours the \glspl{gnb} (i.e., IAB-nodes) had been active, while maintaining a minimum downlink capacity for each \gls{ue} equals to 80Mb/s. This shows how introducing dynamic optimizations enabled by the O-RAN architecture can effectively target improvements in energy efficiency.

\section{System Model and Optimization}\label{sect:problem}
%\subsection{Graph formulation}
Let us first introduce the problem formally. We start from a weighted directed graph $\pazocal G = (\pazocal V, \pazocal E)$, called measurements graph,  whose nodes can be either \gls{iab}-nodes or \glspl{ue}. We denote the set of \glspl{ue} as $\pazocal U \subset \pazocal V$ and the \gls{iab}-donor as $t \in V$.
Each edge $(u,v)$ of this graph represents a potentially usable wireless link between each node and it is weighted by its available capacity ($c(u,v)$), which depends on the channel quality. 
Since the goal is to find a tree representing the routing from each \gls{ue} to towards the donor $t$, the edges of the graph will be directed accordingly. 
Access links (originating from the \glspl{ue}) will always have \glspl{ue} as source and \gls{iab}-nodes as destination. 
Backhaul links involving $t$ will always point towards it, as it is always the destination to reach the core. The links between \gls{iab}-nodes instead can be used in one or the other direction to build the \gls{iab}-tree, so the measurements graph contains a couple of links per each neighbor \gls{iab}-node pair. \cref{subfig:graph} and \cref{subfig:tree} report an example of a measurements graph and a possible \gls{iab}-tree.

Local detailed information on the feasibility of wireless links between \glspl{ue} and \gls{iab}-nodes is available on each \glspl{gnb}. The \gls{oran} architecture allows extensions to standard interfaces so that we can assume that the local information can be collected by an rApp, running on the non-real-time \gls{ric}, which reconstructs the measurements graph we mentioned above. 
Then, the optimization algorithm periodically runs and pushes the optimized topology to the \gls{ran} through the O1 interface. Note that we take into consideration periodic updates of the topology with a period in the order of minutes, so we assume that disabled nodes wake up to receive an updated topology with a similar schedule. This schedule is also perfectly compatible with the non-real-time \gls{ric} closed-loop time constraints.
Without loss of generality in the following model we will assume the optimization of a single tree, but the proposed optimization model can be trivially adapted to optimize multiple trees.

\begin{figure}
    \hfill
    \begin{subfigure}{.3\linewidth}
      \centering
      \begin{tikzpicture}
    \definecolor{ue}{RGB}{0,0,0}
    \definecolor{relay}{RGB}{0,0,255}
    \definecolor{donor}{RGB}{178,34,34}

    \node[circle, draw, thin, fill=ue, scale=0.7] (ue1) at (0, 0) {};
    \node[circle, draw, thin, fill=ue, scale=0.7] (ue2) at (0, 1) {};
    
    \node[circle, draw, thin, fill=relay] (relay1) at (0.5, 0.5) {};

    \node[circle, draw, thin, fill=relay] (relay2) at (1.5, 1.5) {};
    \node[circle, draw, thin, fill=ue, scale=0.7] (ue3) at (0.7, 1.2) {};
    \node[circle, draw, thin, fill=ue, scale=0.7] (ue4) at (2.1, 1.) {};

    \node[circle, draw, thin, fill=relay] (relay3) at (2.3, .3) {};
    \node[circle, draw, thin, fill=ue, scale=0.7] (ue5) at (1.4, 0) {};

    \node[circle, draw, thin, fill=donor,scale=1.3] (donor) at (1.3, 0.9) {};

    \draw [-to](ue1) -- (relay1);
    \draw [-to](ue2) -- (relay1);
    \draw [-to](ue3) -- (relay1);
    \draw [-to](ue3) -- (donor);
    \draw [-to](ue3) -- (relay2);
    \draw [-to](ue4) -- (relay2);
    \draw [-to](ue4) -- (donor);
    \draw [-to](ue4) -- (relay2);
    \draw [-to](ue4) -- (relay3);
    \draw [->](ue5) to[out=180-10, in=-50] (relay1);
    \draw [->](ue5) to[out=-20, in=-100] (relay3);

    \draw [-to](relay1) -- (donor);
    \draw [-to](relay2) -- (donor);
    \draw [-to](relay3) -- (donor);

    \draw [->](relay1) to[out=10, in=-200] (relay3);
    \draw [->](relay3) to[out=180+10, in=180-200] (relay1);

%\begin{scope} [every node/.style = {scale = 0.4}]
%    \matrix [draw=black, row sep=0.2cm,column sep=0.2cm] at (1.2, -0.5)
%  {
%    \node[circle, draw, thin, fill=ue] (ue1) at (0, 0) {}; &
%    \node[left] {UE~~~~~~};
%    \node[circle, draw, thin, fill=relay] (ue1) at (0, 0) {}; &
%    \node[left] {Relay~~~~~~~~~};
%    \node[circle, draw, thin, fill=donor] (ue1) at (0, 0) {}; &
%    \node[left] (e) {Donor};\\
%  };
%\end{scope}
\end{tikzpicture}
        \subcaption[]{}
        \label{subfig:graph}
    \end{subfigure}%
    \hfill
    \begin{subfigure}{.3\linewidth}
      \centering
      \begin{tikzpicture}
    \definecolor{ue}{RGB}{0,0,0}
    \definecolor{relay}{RGB}{0,0,255}
    \definecolor{donor}{RGB}{178,34,34}

    \node[circle, draw, thin, fill=ue, scale=0.7] (ue1) at (0, 0) {};
    \node[circle, draw, thin, fill=ue, scale=0.7] (ue2) at (0, 1) {};
    
    \node[circle, draw, thin, fill=relay] (relay1) at (0.5, 0.5) {};

    \node[circle, draw, thin, fill=relay, opacity=0.2] (relay2) at (1.5, 1.5) {};
    \node[circle, draw, thin, fill=ue, scale=0.7] (ue3) at (0.7, 1.2) {};
    \node[circle, draw, thin, fill=ue, scale=0.7] (ue4) at (2.1, 1.) {};

    \node[circle, draw, thin, fill=relay] (relay3) at (2.3, .3) {};
    \node[circle, draw, thin, fill=ue, scale=0.7] (ue5) at (1.4, .1) {};

    \node[circle, draw, thin, fill=donor,scale=1.3] (donor) at (1.3, 0.9) {};

    \draw [-to](ue1) -- (relay1);
    \draw [-to](ue2) -- (relay1);
    %\draw [-to](ue3) -- (relay1);
    \draw [-to](ue3) -- (donor);
    %\draw [-to](ue3) -- (relay2);
    \draw [-to](ue4) -- (relay3);
    %\draw [-to](ue4) -- (donor);
    %\draw [-to](ue4) -- (relay3);
    %\draw [-to](ue5) -- (relay1);
    \draw [-to](ue5) -- (relay3);

    \draw [-to](relay1) -- (donor);
    %s\draw [-to](relay2) -- (donor);
    \draw [-to](relay3) -- (relay1);

%\begin{scope} [every node/.style = {scale = 0.4}]
%    \matrix [draw=black, row sep=0.2cm,column sep=0.2cm] at (1.2, -0.5)
%  {
%    \node[circle, draw, thin, fill=ue] (ue1) at (0, 0) {}; &
%    \node[left] {UE~~~~~~};
%    \node[circle, draw, thin, fill=relay] (ue1) at (0, 0) {}; &
%    \node[left] {Relay~~~~~~~~~};
%    \node[circle, draw, thin, fill=donor] (ue1) at (0, 0) {}; &
%    \node[left] (e) {Donor};\\
%  };
%\end{scope}
\end{tikzpicture}
      \subcaption[]{}
      \label{subfig:tree}
    \end{subfigure}%
    \hfill~
    \vspace{5pt}
    \caption{Example of a measurements graph $\pazocal G$ (a) and a possible IAB Tree $\pazocal T$ (b). \gls{iab}-donors are depicted in red, \gls{iab}-nodes in blue, \glspl{ue} in black, and deactivated \gls{iab}-nodes in light blue.}\label{fig:graph}
\end{figure}
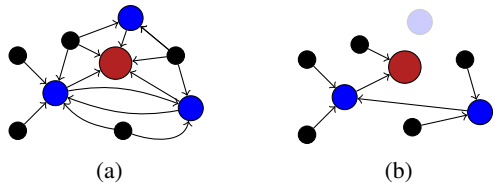

\subsection{Optimization problem}
% The goal of our optimization model is to find a tree $\pazocal T$ such that $\pazocal T$ is a subgraph of $\pazocal G$ rooted in the \gls{iab}-donor and whose leaves are all the \glspl{ue}. $\pazocal T$ should minimize energy consumption. This metric can be defined as the sum of the static energy, consumed even when a \gls{iab}-node is idle and of the dynamic energy, which depends on the number of radio resources the \gls{iab}-node has to serve. However, it has been shown \cite{piovesan2022machine} that the static energy accounts for more than 70\% of the maximum energy consumed by a \gls{gnb}. Hence, since the total energy is largely dominated by static energy, we restrict our study to a simple model where only static energy is considered. This yields an objective which is to find the tree $\pazocal T$ that minimizes the number of activated \gls{iab}-nodes.
Our optimization model aims to identify a tree, denoted as $\pazocal T$, which is a subgraph of $\pazocal G$, rooted in the \gls{iab}-donor, and whose leaves are the \glspl{ue}. The primary objective of $\pazocal T$ is to minimize the energy consumption of the network. This consumption is generally viewed as a combination of static energy, which is expended even when an \gls{iab}-node remains idle, and dynamic energy, which depends on the volume of radio resources the \gls{iab}-node has to serve. As per the findings in \cite{piovesan2022machine}, static energy constitutes over 70\% of the peak energy consumption by a \gls{gnb}. Given this significant skew towards static energy, our model simplifies to focus predominantly on static energy. The resultant objective is to construct the tree $\pazocal T$ such that the activation of \gls{iab}-nodes is minimized.

Additionally, since the whole network operates using the same spectrum, we assume that each node has a \gls{tdma} scheduler that operates using a round-robin policy to schedule the inbound traffic and a dedicated radio device to relay the outbound traffic. This additional constraint---which follows guidance from 3GPP technical documents~\cite{3gpp.38.874}---differentiates our model from a classical multicommodity flow problem, where adjacent edges do not have to share the same time resources as in a wireless network.

We begin the formulation of the problem as a binary multicommodity flow problem. In such a problem, we have to route a set $\pazocal K$ of commodities on the graph, each using a single path. A commodity $k \in \pazocal K$ is defined as a triplet $s_k, t_k, d_k$ where $s_k$ is the source node (in our case, a \gls{ue}), $t_k$ is the destination node (in our case, the \gls{iab}-donor $t$) and $d_k \in \mathbb{R}$ is the bandwidth to reserve on the path from $s_k$ to $t_k$. 
These commodities are decided by the \gls{mno} beforehand, depending on the minimum capacity it wants to guarantee to its customers, and might be differentiated by different classes. The \gls{mno} can feed this information to the rApp running the optimization problem.
We denote by $\pazocal N_{out}(v)$ the outer neighbors of node $v$ and by $\pazocal N_{in}(v)$ its inner neighbors. The cardinality of these sets (\textit{e.g.,} the outer and inner degrees) are denoted by $out(v)$ and $in(v)$, and their sum (the degree of the node) $deg(v) = out(v) + in(v)$.
Let us introduce the binary variables $a(v)~\forall v \in \pazocal V$ which indicate whether node $v$ is turned on or sleeping, binary variables $f_k(u,v)$ which indicate whether commodity $k$ uses edge $(u,v) \in \pazocal E$, and binary variables $f(u,v)$ which indicate whether $(u,v)$ is used by any commodity.
We define the problem as the following binary non-linear programming problem:
\vspace{-1em}

\begin{small}
\allowdisplaybreaks
    \begin{align}
        \min_{}\quad & \sum_{v \in \pazocal{V}} a(v) \label{eq:objective}
        \\
%        \phantom{i + j + k}
  & \begin{aligned}
        \text{s.t.\quad} & \sum_{k \in \pazocal K} f_k(u,v) \cdot d_k  \leq c(u,v) \\&\quad\quad\quad\quad\times \frac{1}{\sum\limits_{w \in \pazocal N_{in}(v)}f(w, v) }& \forall (u,v) \in \pazocal{E}, \forall k \in \pazocal K\label{eq:capacity}    
        \end{aligned}
         \\
        & \sum\limits_{v \in \pazocal V}f_k(u,v) - \sum\limits_{v \in \pazocal V}f_k(v, u) = 0~~~~~~~ \forall u \in \pazocal{V}, \forall k \in \pazocal K \label{eq:equilibrium}\\
        &\sum\limits_{v \in \pazocal V} f_k(s_k,v) - \sum\limits_{v \in \pazocal V} f_k(v, s_k) = 1~~~~~~~~~~~~~~~\forall k \in \pazocal K\label{eq:source}\\
        &\sum\limits_{v \in \pazocal V} f_k(v, t_k)  - \sum\limits_{v \in \pazocal V} f_k(t_k, v) = -1 ~~~~~~~~~~~~~ \forall k \in \pazocal K \label{eq:destination}\\
        &
        \begin{aligned}
            a(v) \geq \frac{1}{deg(v)} \left[ \sum_{\forall u \in \pazocal N_{in}(v)} f(u,v) + \sum_{\forall u \in \pazocal N_{out}(v)} f(v, u) \right] \\ \forall (v) \in \pazocal E \label{eq:activation}
        \end{aligned}\\
        &f(u,v) \geq f_k(u,v) ~~~~~~~~~~~~~~~~~~~ \forall (u,v) \in \pazocal E, k \in \pazocal K \label{eq:aggregation}\\
        &\sum\limits_{v \in \pazocal N_{out}(u)} f(u,v) \leq 1 ~~~~~~~~~~~~~~~~~~~~~~~~~~~~~~ \forall u \in \pazocal V \label{eq:tree}\\
        & a(v), f(u, v), f_i(u,v) \in \{0, 1\} \label{eq:binary}
      \end{align}
\end{small}

Our objective in \cref{eq:objective} is to minimize the number of nodes that are turned on, \textit{e.g.,} the energy consumption of the network. 
In \cref{eq:activation}, the value of variable $a(v)$ is enforced to be $1$ if any flow uses node $v$. \cref{eq:capacity,eq:equilibrium,eq:source,eq:destination} are multi-commodity flow constraints, where \cref{eq:equilibrium,eq:source,eq:destination} enforce the equilibrium of the flow and \cref{eq:capacity} ensures the capacity constraints are respected. This constraint is different from the classic multicommodity flow problem, in which it would be
$$f_k(u,v) \cdot d_k  \leq c(u,v)  ~~~~~~\forall (u,v) \in \pazocal{E}, \forall k \in \pazocal K\label{eq:capacity_classic}.$$
In fact, as mentioned above, in a wireless network the edges adjacent to the same node need to share the spectrum, typically by using \gls{tdma} with a specific scheduler. In our case, we have assumed that a Round Robin scheduler allocates equal resources to all the adjacent edges.
Finally, the constraint in \Cref{eq:aggregation} ensures that an edge is activated if any commodity uses it and \cref{eq:tree} makes sure all activated nodes have outer degree 1, which implies the network is a tree.

This model is non-linear because of the inverse function in Equation \cref{eq:capacity}.
We now propose an equivalent linearized version of the previous model. We prove the equivalence in Theorem 1.

In the linearized model below, we introduce binary variables $x_i(v)~~\forall v \in \pazocal V$. These variables are equal to 1 iff at least $i$ of the inner edges incident to $v$ are activated. This enables us to linearize the inverse function in \cref{eq:capacity} and to replace it with a weighted sum of those binary variables.
\begin{small}
\allowdisplaybreaks
\begin{align}
    \min_{}\quad & \sum_{v \in \pazocal{V}} \sum_{i = 1}^{in(v)} x_i(v) \label{eq:objective2} \\
    & \begin{aligned}
        \text{s.t.\quad} & f(u,v) \cdot d_k  \leq c(u,v) \cdot \left ( x_1(v) - \sum_{i=2}^{in(v)} \frac{x_i(v)}{(i-1)i} \right ) \\&\quad\quad\quad\quad\quad\quad\quad\quad\quad\quad\quad\quad\quad\quad\quad~~ \forall (u,v) \in \pazocal{E} \label{eq:capacity2}    
        \end{aligned}
     \\
    % % \cdot \left [ x_1(v) - \frac{1}{2} x_2(v) - (\frac{1}{2} \frac{1}{3}) x_3(v) - ... (\frac{1}{d^+-1(v)} \frac{1}{d^+(v)}) x_{d^+(v)} \right ]~~~& 
    &\begin{aligned}
        x_i(v) \geq \left(\sum_{u \in \pazocal N_{in}(v)}f(u, v) - (i-1) \right) / in(v) \\ ~~~~~~~~~~~~~~~~~~~~~~~~~~~~~~~~~~~~\forall v \in \pazocal V, \forall 1 \leq i \leq d(v)
    \end{aligned}\label{eq:degree}\\
    & x_i(v) \in \{0, 1\} \quad\quad\quad\quad\quad\quad~~\forall v \in \pazocal V, \forall 1 \leq i \leq d(v)\label{eq:binary2}\\
    & (\ref{eq:equilibrium}), (\ref{eq:source}), (\ref{eq:destination}),  (\ref{eq:aggregation}), (\ref{eq:tree}), (\ref{eq:binary}) \nonumber
\end{align}
\end{small}

\begin{theorem}
The BNLP (\ref{eq:objective}) - (\ref{eq:binary}) has the same optimal solution as the BLP (\ref{eq:equilibrium}) - (\ref{eq:binary2}).
\end{theorem}
\begin{proof}
Let us first observe that $\left(\sum\limits_{u \in \pazocal N_{in}(v)}f(u, v) - (i-1) \right)$ is always positive if at least $i$ inner edges of $v$ are activated, and is nonpositive otherwise. also note that this sum is always lower or equal to $in(v)$. Hence, the right-hand side of \cref{eq:degree} is between -1 and 1, and its value is positive if $i$ inner edges are used. This, combined with the fact we are minimizing the sum of variables $x_i(v)$ means that in an optimal solution to problem $(\ref{eq:equilibrium}) - (\ref{eq:binary2})$, $x_i(v)$ will be equal to 1 if at least $i$ inner edges of $v$ are activated and $0$ otherwise.
~\\
Let us now observe that in \cref{eq:capacity2}, if $n$ inner edges of $v$ are activated, then $x_1(v), x_2(v), ... x_n(v)$ will be equal to 1. It follows that the sum $x_1(v) - \sum_{i=2}^{in(v)} \frac{x_i(v)}{(i-1)i}$ will be equal to $\frac{1}{n}$, \textit{e.g.} the constraint is equivalent to constraint (\ref{eq:capacity}).
Finally, observe that since we are building a tree, minimizing its number of edges is equivalent to minimizing its number of nodes, as a tree of $n$ nodes always has exactly $n-1$ edges, meaning the objective function, \Cref{eq:objective2}, is equivalent to the objective in \Cref{eq:objective}.

\end{proof}

\section{Performance Evaluation Setup}\label{sect:setup}
This section presents the techniques used to synthetically generate the set of measurements graphs $\Ggraph(\Vset, \Eset)$, needed to evaluate the feasibility and effectiveness of our optimization model. In particular, we will be using datasets representing an area of 0.092km$^2$ in the center of Milan, Italy. To do so, in the first subsection, we will describe the state-of-the-art techniques used to place the \gls{iab}-nodes \cite{gemmi2022costeffective}, and the \glspl{ue} \cite{3gpp.38.913}. Then, in the second subsection, we present our channel model, based on 3GPP specifications combined with ray tracing. Finally, the third section presents our data-driven time-varying UE density model, which enables us to generate different instances depending on the time of the day and the day of the week.

\subsection{Placement of gNBs and UEs}
The set of nodes of our graph \Vset is comprised of both \gls{iab}-nodes, and \glspl{ue}, whose placement is done separately using two different techniques. \gls{iab}-nodes are placed on building facades with a given density \lambdagnb. 
The exact position is computed by taking advantage of a state-of-the-art placement heuristic \cite{gemmi2022costeffective} that exploits highly precise 3D models to place the gNB such that the number of potential \glspl{ue}' location in line of sight is maximized.
\glspl{ue} are then randomly distributed both in public areas, such as streets, and inside buildings. Specifically, given a density of \lambdaue, indoor \glspl{ue} are uniformly randomly distributed inside buildings with a density equal to $r_{i/o}\cdot\lambdaue$ and outdoor \glspl{ue} are uniformly randomly distributed inside buildings with density $(1-r_{i/o}) \cdot\lambdaue$, where $r_{i/o}$ is a commonly used ratio of indoor to outdoor \gls{ue} equal to 0.8 taken from \gls{3gpp} technical report \cite{3gpp.38.913}. In short, we consider that in our simulations 80\% of the UEs are placed indoors.
\cref{fig:map} shows a deployment with $\lambdagnb=45$ and $\lambdaue=900$ \gls{ue}/km$^2$.

\begin{figure}
    \centering
    \includegraphics[width=0.6\linewidth]{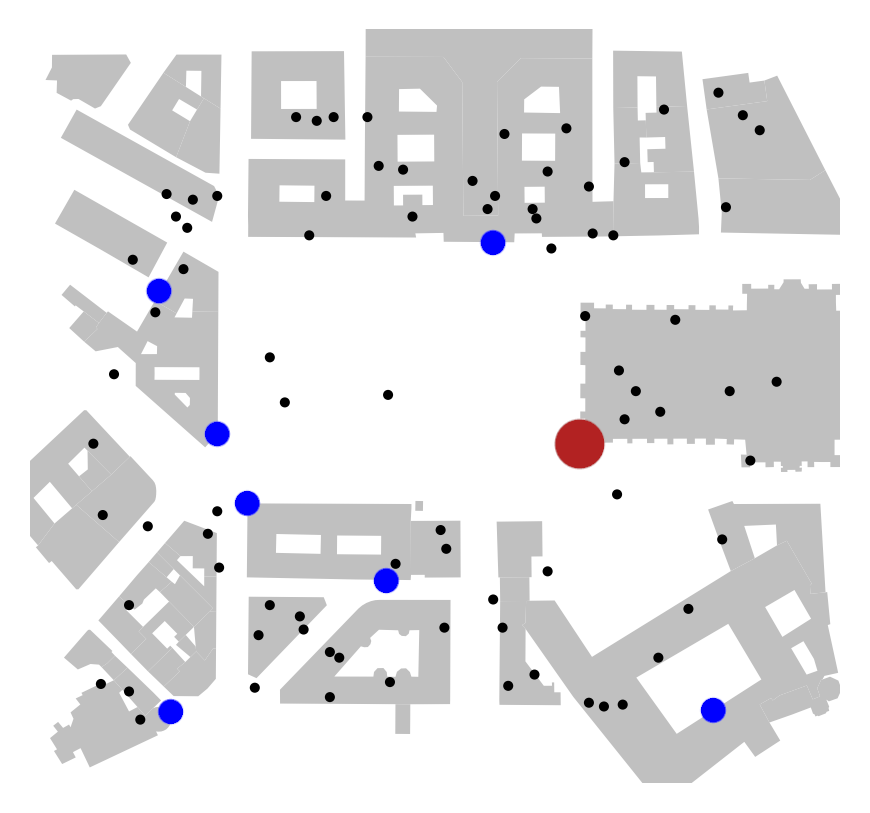}
    \caption{Sample deployment of a network in the center of Milan, with 1 \gls{iab}-donor (in red), 7 \gls{iab}-nodes (in blue), and 83 \glspl{ue} (in black). It corresponds to $\lambdaue(9)=900$ \gls{ue}/km$^2$ (Mon 9am).}
    \label{fig:map}
\end{figure}

\subsection{Access and Backhaul channel models}
Once the location of both \glspl{ue} and \gls{iab}-nodes have been determined,  we evaluate the path loss by applying the \gls{3gpp} Urban-Micro (UMi) stochastic channel model \cite{3gpp.38.901}.  However, instead of using the stochastic \gls{los} probability model provided by the same UMi model, we deterministically evaluate the \gls{los} by employing ray tracing analysis on the same 3D models used to find the optimal locations, obtaining a more accurate estimation  \cite{gemmi2022properties}. 
For indoor \glspl{ue}, we always consider them to be \gls{nlos} and we add the additional \gls{o2i} penetration loss. Since the buildings in the area we consider are mostly made out of concrete, we use the high-loss \gls{o2i} model \cite{3gpp.38.901}.

Finally, we compute the \gls{snr} using the thermal noise and by adding the receiver noise figure, then we calculate the Shannon capacity. Both access and backhaul are assumed to be using the same frequencies, but different values of antenna gain and numbers of MIMO layers are used. 
\cref{tab:simulation_params} details all the values used in our simulations, which are aligned with typical literature and 3GPP studies on this topic.

\subsection{Time-varying \gls{ue} density model}
As explained in more detail in \Cref{sect:intro}, most studies dealing with topology optimization focus their analysis on a single, or a handful, value of $\lambdaue$. 
Since the energy optimization technique we devise tunes the \gls{iab}-node activation on the basis of the number of \glspl{ue} and their load, we need to evaluate our model on a large number of values of \gls{ue} density, ideally following a realistic trend. 
Therefore, we employ a technique used in similar research \cite{baiocchi2017joint} to devise a time-varying \gls{ue} density model. 
First, we extract the cell load profile $p(t)$ related to our analysis area, in Milan,  from openly available datasets \cite{barlacchi2015multi}. We then normalize it in the range $(0,1]$, and we model the  \gls{ue} density as a function of time $\lambdaue(t) = p(t)l\lambdagnb $, where $l=10$ is the number of \glspl{ue} per \glspl{gnb} taken from the \gls{3gpp} technical report~\cite{3gpp.38.913}.
Finally, we generate a set of 168 graphs spanning an average week with a one-hour granularity. 
%Note that since the \gls{iab}-node placement is not random, but deterministic, the only difference between each graph will be in the number and position of \glspl{ue}, and their channels.

\cref{fig:ues} reports the hourly trend of $\lambda_{\gls{ue}}(t)$ corresponding to the area of our analysis, showing how for several hours every night the network has to serve almost no \glspl{ue} and how in the weekends, even at peak hours, the density of \glspl{ue} never exceed 80\% of the weekday peak hours.

\begin{figure}[]
    \centering
    \begin{tikzpicture}

  \definecolor{c1}{RGB}{213,94,0}
  \definecolor{c2}{RGB}{1,115,178}
  \definecolor{c3}{RGB}{2,158,115}
  \definecolor{c4}{RGB}{176,176,176}
  \definecolor{c5}{RGB}{222,143,5}
  \definecolor{c6}{RGB}{204,204,204}
  \definecolor{c7}{RGB}{204,120,188}
  \pgfplotstableset{col sep=comma}
  \begin{axis}[
  width=.9\linewidth,
  height=.55\linewidth,
  xticklabel style = {align=center, font=\scriptsize},
  yticklabel style = {align=center, font=\scriptsize},
  ylabel style ={font=\footnotesize},
  xlabel style ={font=\footnotesize},
  legend cell align={left},
  legend style={draw opacity=1, text opacity=1, draw=c6, fill opacity=0.8, at={(0.3,0.5)},anchor=west},
  tick align=outside,
  tick pos=left,
  x grid style={c4},
  xlabel={t ~[h]},
  xmajorgrids,
  xmin=0, xmax=167,
  xtick style={color=black},
  y grid style={c4},
  ylabel={$\lambda_{\text{UE}}(t)$},
  ymajorgrids,
  ymin=0, ymax=910,
  ytick style={color=black},
  y filter/.code={\pgfmathparse{#1*900}\pgfmathresult}
  ]
  \addplot+[c2, very thick, mark=None] table [x=time, y=call-in] {data/ue_calls.csv};
  \end{axis}
  
  \end{tikzpicture}
  
    \caption{Weekly profile for the \gls{ue} density in central Milan.}
    \label{fig:ues}
\end{figure}
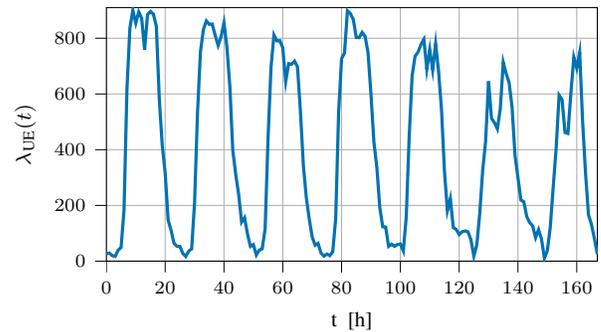

\section{Results}\label{sect:results}

\begin{table}
    \centering
    \begin{tabular}{ll}
        \toprule
        Parameter & Value \\
        \midrule
        Area size & 0.092 km$^2$ \\
        %gNB Density & 90 gNB/km$^2$ \\
        \gls{ue} density range & [0-900] \gls{ue}/km$^2$ \\
        Indoor/Outdoor \gls{ue} ratio & 80/20 \\
        Carrier frequency & 28 GHz \\
        Bandwidth & 100 MHz \\
        Noise Figure & 5 dB \\
        %3GPP Channel Model & Urban Micro \\
        O2I Loss & 14.15 dB \\ 
        Reception gain (Access/Backhaul) & 3 / 10 dBi \\
        MIMO layers (Access/Backhaul) & 2 / 4 \\
        Backhaul transmission power & 30 dBm \\
        Minimum Capacity per \gls{ue} & 80Mb/s \\
        Number of independent simulation runs & 10 \\
        \bottomrule
    \end{tabular}
    \caption{Simulation Parameters} \label{tab:simulation_params}
\end{table}
As already mentioned in \cref{sect:intro}, we evaluate our model on a 0,092km$^2$ area in the center of the city of Milan (Italy), for which we computed the \gls{ue} density trend $\lambdaue$ of an average week. For each hour of the week (168 in total), we generate the measurements graph as described in the previous section and then we run our optimization algorithm on it. We compare the trees found by our solution with 4 strategies:
\begin{itemize}
    \item \textbf{All donors}, a dense deployment without \gls{iab}, where all the \gls{gnb} are wired. It is an upper bound in terms of energy consumption and capacity. Additionally, no re-distribution of the \glspl{ue} is performed as they are always attached to the \gls{gnb} with the lowest \gls{snr}.
    \item \textbf{No relays}, a deployment where all the \gls{iab}-nodes are not active. It is a lower bound in terms of energy and capacity.
    \item \textbf{Widest Tree}, a strategy that employs the well-known widest path algorithm to find the path of maximum capacity (\textit{e.g.,} with the largest bottleneck in terms of capacity) from each \gls{ue} towards the donor and deactivates all the \gls{iab}-nodes that are not part of any path.
    \item \textbf{Optimized Tree}, our optimization model.  % \cite{ma1997path}
\end{itemize}
In the first part of this section, we compare the energy consumption (both in terms of the number of nodes activated and of the overall number of gNB-hours) of the different algorithms. Then, in the second part, we evaluated the topologies in terms of bottlenecks of the downlink capacity.

%First, we evaluate the number of active \gls{iab}-nodes, to measure the energy consumption of the different strategies, then we evaluate analytically evaluate the \gls{iab} topology to find the bottlenecks between each \gls{ue} and the \gls{iab}-donor, in terms of downlink capacity. 
%The same performance analysis evaluation is also performed on other \gls{iab} topologies obtained by different strategies using the same measurement graph:

\subsection{Energy Consumption}
To evaluate the energy consumption of the \gls{iab} networks we first show the hourly number of active \gls{iab}-nodes, then we introduce a metric that measures the total number of hours each \gls{iab}-node has been active. The \gls{iab}-donor is not taken into consideration as we always need at least one node to be active to provide a minimum service to the users. 

Figure \ref{fig:activated} shows the number of \gls{iab}-nodes activated, on the left axis, and the number of \glspl{ue} connected to the network, on the right axis. To improve the readability only the values for the first day of the week have been reported.
\textbf{Optimized Tree}, shows that it is possible to fully deactivate the \gls{iab}-nodes at nighttime (from 12 pm to 4 am) and that also during daytime several \gls{iab}-nodes can be deactivated. By comparing its trend with the number of \glspl{ue}, we can also see that it gets perfectly followed, highlighting the effectiveness of our optimization model.
\textbf{Widest Tree}, on the other hand, never manages to deactivate more than 3 \gls{iab}-nodes, highlighting that a specific algorithm is needed to fully implement energy-saving policies.

\begin{figure}
    \centering
    \begin{tikzpicture}

  \definecolor{c1}{RGB}{213,94,0}
  \definecolor{c2}{RGB}{1,115,178}
  \definecolor{c3}{RGB}{2,158,115}
  \definecolor{c4}{RGB}{176,176,176}
  \definecolor{c5}{RGB}{222,143,5}
  \definecolor{c6}{RGB}{204,204,204}
  \definecolor{c7}{RGB}{204,120,188}
  \pgfplotstableset{col sep=comma}
  \begin{axis}[
  width=.85\linewidth,
  height=0.6\linewidth,
  ylabel near ticks,
  xticklabel style = {align=center, font=\scriptsize},
  yticklabel style = {align=center, font=\scriptsize},
  ylabel style ={font=\footnotesize},
  xlabel style ={font=\footnotesize},
  legend cell align={left},
  legend columns=3,
  legend style={draw opacity=1, text opacity=1, draw=c6, fill opacity=0.8, at={(0.3,0.5)},anchor=west, font=\footnotesize, at={(0.5,1.03)}, anchor=south},
  tick align=outside,
  tick pos=left,
  x grid style={c4},
  xlabel={Time [h]},
  xmajorgrids,
  xmin=0, xmax=23,
  xtick style={color=black},
  y grid style={c4},
  ylabel={Activated \gls{iab}-nodes},
  ymajorgrids,
  ymin=-0.2, ymax=7.2,
  ytick style={color=black},
  ytick = {0,1,2,3,4,5,6,7},
  axis y line*=left,
  error bars/y dir=both, % turn on error bars
  error bars/y explicit,  % say that error value is given explicitly
  ]
  \addplot+[c7, very thick, dash pattern=on 5pt off 1pt on 1pt off 1pt, mark options={solid, scale=0.7}, mark=square] table [x=time, y expr= {\thisrowno{5}-1}, y error=Widest Tree_ci] {data/activated.csv};
  \addlegendentry[]{Widest Tree}
  \addplot+[c3, very thick, mark=diamond, mark options={solid, scale=0.7}] table [x=time, y expr= {\thisrowno{3}-1}, y error=Optimized Tree_ci] {data/activated.csv};
  \addlegendentry[]{Optimized Tree}
  \addlegendimage{/pgfplots/refstyle=plot_ue}\addlegendentry{\# \glspl{ue}}
  \end{axis}

\begin{axis}[
  width=.85\linewidth,
  height=0.6\linewidth,
  ylabel near ticks,
  xticklabel style = {align=center, font=\scriptsize},
  yticklabel style = {align=center, font=\scriptsize},
  ylabel style ={font=\footnotesize},
  xlabel style ={font=\footnotesize},
  legend cell align={left},
  legend columns=2,
  legend style={draw opacity=1, text opacity=1, draw=c6, fill opacity=0.8, at={(0.3,0.5)},anchor=west, font=\footnotesize, at={(0.5,1.03)}, anchor=south},
  tick align=outside,
  tick pos=right,
  xmin=0, xmax=23,
  xtick style={color=black},
  y grid style={c4},
  ylabel={\# \glspl{ue}},
  ymin=-2.9, ymax=108,
  ytick style={color=black},
  ytick = {0,15,30,45, 60,75, 90,105},
  axis y line*=right,
  axis x line=none,
  ]
  \addplot+[c1, very thick, mark options={solid}, mark=+, dash pattern=on 1pt off 1pt] table [x=time, y expr= {\thisrowno{1}*83}] {data/ue_calls.csv};\label{plot_ue}
  % \addlegendentry[]{# \glspl{ue}}
  \end{axis}

  \end{tikzpicture}
  
    \caption{Number of \gls{iab}-nodes activated in the first 24 of the week (Monday).}
    \label{fig:activated}
\end{figure}
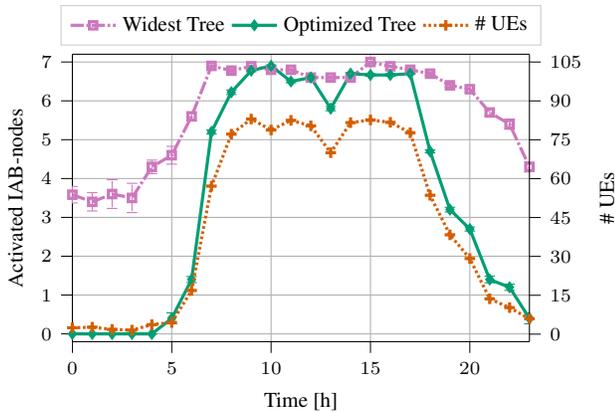

Additionally, by integrating the number of activated \gls{gnb} at each hour for the span of the week we obtain the total number of \gls{gnb}-hours for each strategy. For \textbf{No relays}, the value is $168h$, as only one \gls{gnb} is always active. For \textbf{All donors}, on the other hand, the total number of \gls{gnb}-hours is equals to $168h\cdot8=1344h$, since 8 \gls{iab}-donors are active at all times. More interestingly, the values for \textbf{Widest Tree} and \textbf{Optimized Tree} respectively activate the \gls{ran} for $1141h$ and $709h$, which means our method improves the power consumption of 47\% over \textbf{All Donors} and 38\% over \textbf{Widest Tree}.

\subsection{Capacity}
To evaluate the performance of the topology we analyze the capacity served to each \glspl{ue} with three different capacity metrics, which are shown in \cref{fig:capacity} and detailed below. First, let us define some functions used throughout the section. Let $p(u,t)$ be the function returning the set of edges forming the path from $u$ to $t$ over our topology tree. Let $N_{in}(t)$ the number of edges directed towards $t$ in the tree and $c(i,j)$ the capacity of the edge $(i,j)$.

\begin{figure*}[h]
    \subfloat[Average idle capacity. \label{fig:max_capacity}]{
        % This file was created with tikzplotlib v0.10.1.
\begin{tikzpicture}

  \definecolor{chocolate213940}{RGB}{213,94,0}
  \definecolor{darkcyan1115178}{RGB}{1,115,178}
  \definecolor{darkcyan2158115}{RGB}{2,158,115}
  \definecolor{darkgray176}{RGB}{176,176,176}
  \definecolor{darkorange2221435}{RGB}{222,143,5}
  \definecolor{lightgray204}{RGB}{204,204,204}
  \definecolor{orchid204120188}{RGB}{204,120,188}
  
  \begin{axis}[
  width=.33\linewidth,
  xticklabel style = {align=center, font=\scriptsize},
  yticklabel style = {align=center, font=\scriptsize},
  ylabel style ={font=\footnotesize},
  xlabel style ={font=\footnotesize},
  legend cell align={left},
  legend columns = 2,
  legend style={draw opacity=1, text opacity=1, draw=lightgray204, fill opacity=0.8, font=\scriptsize, at={(0.5,1.03)}, anchor=south},
  log basis y={10},
  tick align=outside,
  tick pos=left,
  x grid style={darkgray176},
  xlabel={Time [h]},
  xmajorgrids,
  xmin=0, xmax=23,
  xtick style={color=black},
  y grid style={darkgray176},
  ylabel={$\hat c_I$ [Mb/s]},
  ymajorgrids,
  ymin=0.9, ymax=2.5,
  ymode=linear,
  xtick={0,5,10,15,20},
  ytick style={color=black}
  ]
  \addplot+[thick, darkcyan1115178, mark=o, mark size=3, mark options={solid,scale=0.7}, dash pattern=on 1pt off 1pt] table [x=time, y expr={\thisrowno{1}/1000}, col sep =comma] {data/max.csv};
  \addlegendentry[]{All donors}
  \addplot+[thick, darkorange2221435, dash pattern=on 4pt off 1.5pt, mark=x, mark size=3, mark options={solid,scale=0.7}] table [x=time, y expr={\thisrowno{2}/1000}, col sep =comma] {data/max.csv};
  \addlegendentry[]{No relay}
  \addplot+[thick, darkcyan2158115, mark=diamond, mark options={solid,scale=0.7}] table [x=time, y expr={\thisrowno{3}/1000}, col sep =comma] {data/max.csv};
  \addlegendentry[]{Optimized Tree}
  % % \addplot+[semithick, chocolate213940, dash pattern=on 3pt off 1.25pt on 1.5pt off 1.25pt, mark=+, mark size=3, mark options={solid}] table [x=time, y=Voronoi_mean, col sep =comma] {data/max.csv};
  % % \addlegendentry[]{Voronoi}
  \addplot+[thick, orchid204120188, dash pattern=on 5pt off 1pt on 1pt off 1pt, mark options={solid,scale=0.7}, mark=square] table [x=time, y expr={\thisrowno{5}/1000}, col sep =comma] {data/max.csv};
  \addlegendentry[]{Widest Tree}
  \end{axis}
  
  \end{tikzpicture}
  
    }
    \subfloat[Average saturation capacity. \label{fig:min_capacity}]{
        %\raisebox{15px}
        {
        % This file was created with tikzplotlib v0.10.1.
\begin{tikzpicture}

  \definecolor{chocolate213940}{RGB}{213,94,0}
  \definecolor{darkcyan1115178}{RGB}{1,115,178}
  \definecolor{darkcyan2158115}{RGB}{2,158,115}
  \definecolor{darkgray176}{RGB}{176,176,176}
  \definecolor{darkorange2221435}{RGB}{222,143,5}
  \definecolor{lightgray204}{RGB}{204,204,204}
  \definecolor{orchid204120188}{RGB}{204,120,188}
  
  \begin{axis}[
    width=.33\linewidth,
%    height=0.65\linewidth,
  xticklabel style = {align=center, font=\scriptsize},
  yticklabel style = {align=center, font=\scriptsize},
  ylabel style ={font=\footnotesize},
  xlabel style ={font=\footnotesize},
  legend cell align={left},
  legend columns = 2,
  legend style={draw opacity=1, text opacity=1, draw=lightgray204, fill opacity=0.8, font=\scriptsize, at={(0.5,1.03)}, anchor=south},
  log basis y={10},
  tick align=outside,
  tick pos=left,
  x grid style={darkgray176},
  xlabel={Time [h]},
  xmajorgrids,
  xmin=0, xmax=23,
  xtick style={color=black},
  y grid style={darkgray176},
  ylabel={$\hat c_S$ [Gb/s]},
  ymajorgrids,
  ymin=1.8, ymax=4000,
  ymode=log,
  xtick={0,5,10,15,20},
  ytick style={color=black}
  ]
  \addplot+[thick, darkcyan1115178, mark=o, mark size=3, dash pattern=on 1pt off 1pt,, mark options={solid,scale=0.7}] table [x=time, y=All donors, dash pattern=on 1pt off 1pt, col sep =comma] {data/min_mean.csv};
  \addlegendentry[]{All donors} 
  \addplot+[thick, darkorange2221435, dash pattern=on 4pt off 1.5pt, mark=x, mark size=3, mark options={solid,scale=0.7}] table [x=time, y=No relay, col sep =comma] {data/min_mean.csv};
  \addlegendentry[]{No relays}
  \addplot+[thick, darkcyan2158115, mark=diamond, mark options={solid,scale=0.7}] table [x=time, y=Optimized Tree, col sep =comma] {data/min_mean.csv};
  \addlegendentry[]{Optimized Tree}
  % \addplot+[semithick, chocolate213940, dash pattern=on 3pt off 1.25pt on 1.5pt off 1.25pt, mark=+, mark size=3, mark options={solid}] table [x=time, y=Voronoi, col sep =comma] {data/min_mean.csv};
  % \addlegendentry[]{Voronoi}
  \addplot+[thick, orchid204120188, dash pattern=on 5pt off 1pt on 1pt off 1pt, mark options={solid,scale=0.7}, mark=square] table [x=time, y=Widest Tree, col sep =comma] {data/min_mean.csv};
  \addlegendentry[]{Widest Tree}
  \end{axis}
  
  \end{tikzpicture}
  }
    }
    \subfloat[Minimum saturation capacity.\label{fig:minmin_capacity}]{
        %\raisebox{15px}
        {
        % This file was created with tikzplotlib v0.10.1.
\begin{tikzpicture}

  \definecolor{chocolate213940}{RGB}{213,94,0}
  \definecolor{darkcyan1115178}{RGB}{1,115,178}
  \definecolor{darkcyan2158115}{RGB}{2,158,115}
  \definecolor{darkgray176}{RGB}{176,176,176}
  \definecolor{darkorange2221435}{RGB}{222,143,5}
  \definecolor{lightgray204}{RGB}{204,204,204}
  \definecolor{orchid204120188}{RGB}{204,120,188}
  
  \begin{axis}[
    width=.33\linewidth,
    %height=0.65\linewidth,
  xticklabel style = {align=center, font=\scriptsize},
  yticklabel style = {align=center, font=\scriptsize},
  ylabel style ={font=\footnotesize},
  xlabel style ={font=\footnotesize},
  legend cell align={left},
  legend columns = 2,
  legend style={draw opacity=1, text opacity=1, draw=lightgray204, fill opacity=0.8, font=\scriptsize, at={(0.5,1.03)}, anchor=south},
  log basis y={10},
  tick align=outside,
  tick pos=left,
  x grid style={darkgray176},
  xlabel={Time [h]},
  xmajorgrids,
  xmin=0, xmax=23,
  xtick style={color=black},
  y grid style={darkgray176},
  ylabel={$\bar c_S $ [Mb/s]},
  ymajorgrids,
  ymin=1.8, ymax=4000,
  xtick={0,5,10,15,20},
  ymode=log,
  ytick style={color=black}
  ]
  \addplot+[thick, darkcyan1115178, mark=o, mark size=3, mark options={solid,scale=0.7}, dash pattern=on 1pt off 1pt] table [x=time, y=All donors, col sep =comma] {data/min.csv};
  \addlegendentry[]{All donors}
  \addplot+[thick, darkorange2221435, dash pattern=on 4pt off 1.5pt, mark=x, mark size=3, mark options={solid,scale=0.7}] table [x=time, y=No relay, col sep =comma] {data/min.csv};
  \addlegendentry[]{No relays}
  \addplot+[thick, darkcyan2158115, mark=diamond, mark options={solid,scale=0.7}] table [x=time, y=Optimized Tree, col sep =comma] {data/min.csv};
  \addlegendentry[]{Optimized Tree}
  % \addplot+[semithick, chocolate213940, dash pattern=on 3pt off 1.25pt on 1.5pt off 1.25pt, mark=+, mark size=3, mark options={solid}] table [x=time, y=Voronoi, col sep =comma] {data/min.csv};
  % \addlegendentry[]{Voronoi}
  \addplot+[thick, orchid204120188, dash pattern=on 5pt off 1pt on 1pt off 1pt, mark options={solid,scale=0.7}, mark=square] table [x=time, y=Widest Tree, col sep =comma] {data/min.csv};
  \addlegendentry[]{Widest Tree}
  \end{axis}
  
  \end{tikzpicture}
  }
    }
    \vspace{1em}
    \caption{Capacity Metrics for the first 24h of the week (Mon).}\label{fig:capacity}    
\end{figure*}
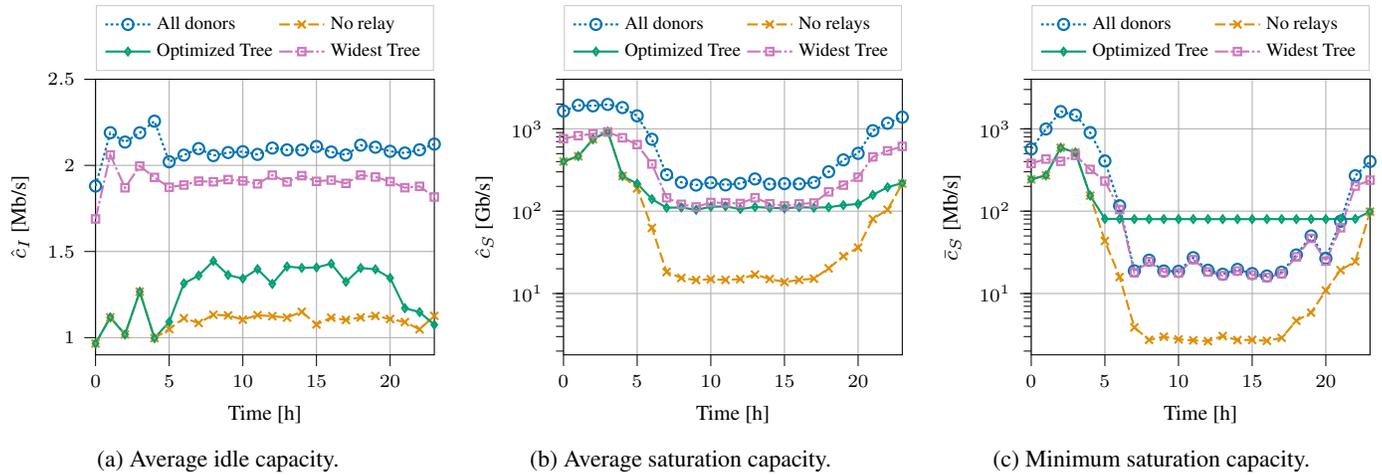

The first metric, called Average Idle Capacity measures the average theoretical capacity of \glspl{ue}, \textit{e.g.} the capacity that would be attainable if the network resources were completely unused. As detailed in \cref{eq:average_idle} below, it is computed as the minimum capacity (the bottleneck) of the edges over the path between each \gls{ue} ($u$) and the donor ($t$). Which is then averaged across all the \glspl{ue} $u \in \pazocal U$. This metric represents an upper bound on the capacity per \gls{ue}.
\begin{equation}\label{eq:average_idle}
\hat c_I = \frac{1}{|\pazocal U|} \sum_{u\in \pazocal U} \min_{(i,j) \in
p(u, t)} c(i,j)
\end{equation}

Figure \ref{fig:max_capacity} shows the Average Idle Capacity for the four different strategies on the first 24 hours of our week. The first insight provided by this figure is that the maximum capacity per \gls{ue} does not depend on the load of the network.
%comment on the first 5 data points.
Moreover, as we were expecting \textbf{All donors} and \textbf{No relay} are respectively the upper and lower bounds in terms of capacity. \textbf{Widest Tree}, the strategy that maximizes the bottleneck between each \gls{ue} and the donor, manages to 
achieve a capacity very close to the upper bound (8\% lower). \textbf{Optimized Tree}  instead shows a more significant drop with a loss of 35\%. 
The drop can be explained by the minimum capacity constraint that, instead of letting each \glspl{ue} reach the \gls{iab}-donor through the widest path, in certain cases picks paths worse in terms of maximum capacity that instead guarantee the minimum capacity.

The second metric, called Average Saturation Capacity and detailed in \Cref{eq:average_saturation}, is formulated in a very similar way as \Cref{eq:average_idle}. However, here we assume that all the \glspl{ue} try to access the network at the same time, thus we divide the capacity of each edge $c_{s,t}$ by the number of inner neighbors of the node $t$, since those edges share the same resources through the scheduler.
\begin{equation}\label{eq:average_saturation}
\hat c_S = \frac{1}{|\pazocal U|} \sum_{u\in \pazocal U} \min_{(i,j) \in p(u, t)} \frac{c(i,j)}{N_{in}(j)}
\end{equation}

As in the previous metric, also here \textbf{All donors} and \textbf{No relays} behaves respectively as upper and lower bound. The difference between the two other strategies, and their distance from the upper bound drops sharply. In fact, at peak time \textbf{All donors} is capable of delivering roughly 200Mb/s per \gls{ue}, while \textbf{Optimized Tree} and \textbf{Widest Tree} respectively deliver 115 and 130 Mb/s per \gls{ue}.

The third metric, called Minimum Saturation Capacity and detailed in \Cref{eq:minimum_saturation}, measures the capacity delivered to worst \gls{ue} while the network is under saturation by all the \glspl{ue}, \textit{e.g.} it defines the minimum level of Quality-of-Service provided by the topology. It is defined similarly to the previous one, but instead of averaging over the \glspl{ue} we take the worst value. 

\begin{equation}\label{eq:minimum_saturation}
    \bar c_S = \min_{u\in \pazocal U} \min_{(i,j) \in p(u, t)} \frac{c(i,j)}{N_{in}(j)}
\end{equation}

\cref{fig:min_capacity} shows that \textbf{Optimized Tree} is the only strategy that manages to guarantee the minimum level of service, equal to 80Mb/s during peak hours (7 am-16 pm), while also minimizing the excessive capacity at night time.
In comparison, with \textbf{No relays} we measure a minimum level of service that, at peak time, is one order of magnitude lower than the minimum level of service (between 2 and 7 Mb/s) while \textbf{Widest Tree} and \textbf{All donors} behave similarly in terms of minimum capacity, as they can both take advantage of all the \gls{iab}-nodes available. 
However, since the \glspl{ue} are not load-balanced across all the available \gls{gnb}, the minimum level of service is not met.
We also note that with \textbf{All donors} there is also an excess of capacity at night time; when energy-saving policies could deactivate several \gls{iab}-nodes, moreover despite being in a more favorable position where no routing is to be performed beyond the first link between the UEs and the BS, \textbf{All donors} still has less capacity than the Optimized Tree. This emphasizes the importance of balancing the load of UEs between base-stations

\section{Conclusions}
In this paper, we have modeled and solved the problem of finding an IAB topology optimized for energy efficiency. We find optimal solutions to the problem, which enables us to save up to 47\% of energy compared to the baseline, while still respecting capacity constraints, which other approaches cannot do. Our results hence show the importance of considering energy-efficiency as a core feature of IAB topology design. In the future, we plan to improve the speed of our algorithm and to propose fast heuristics inspired by it. 
We also plan on evaluating the energy savings and capacity on a real testbed.

\bibliographystyle{IEEEtran}
\bibliography{globecom23}

\end{document}